%% file: llm_bayesian_agentic_ai_risk.tex
\documentclass[11pt]{article}
\usepackage[margin=1in]{geometry}
\usepackage{amsmath,amssymb,amsthm,mathtools,bm}
\usepackage{booktabs}
\usepackage{graphicx}
\usepackage{hyperref}
\usepackage{algorithm}
\usepackage{algpseudocode}
\usepackage[numbers]{natbib}
\usepackage{xcolor}
\usepackage{enumitem}
\usepackage{array}

\newtheorem{definition}{Definition}
\newtheorem{assumption}{Assumption}
\newtheorem{proposition}{Proposition}

\title{\textbf{Belief at Risk}: Quantifying Agentic AI Model Risk with LLM-Inferred Bayesian State Filters\footnote{Python source code available at github.com/mfrdixon/belief-at-risk}}
\author{Matthew Dixon\\Quiota LLC}
\date{\today}

\begin{document}
\maketitle

\begin{center}
\fbox{\begin{minipage}{0.93\textwidth}
\textbf{Executive summary cover sheet for CROs, technology leaders, and portfolio managers.}

Agentic AI introduces a form of model risk that is closer to stochastic control than to ordinary prediction. A conventional model produces a score or forecast. An agent forms beliefs, invokes tools, selects actions, and updates its future information set through the consequences of those actions. This paper proposes a quantitative governance framework in which the central object is not a prompt, benchmark score, or qualitative risk rating, but a time-indexed posterior belief state over latent operating regimes.

The proposed methodology uses a large language model as a probabilistic observation model. The LLM converts structured and semantic evidence into a probability vector over pre-specified latent regimes. A Bayesian filter then regularizes that evidence through a Markov transition prior, producing auditable belief states, entropy, drift, calibration diagnostics, and loss-based tail-risk measures. In this design, the LLM is not granted authority to make portfolio decisions. It supplies uncertain evidence. The Bayesian layer turns that evidence into a controlled model-risk object.

For risk managers, the main output is a residual-risk framework that links uncertainty, decision sensitivity, consequence severity, autonomy, and controls. For technology leaders, it gives a validation architecture for LLM agents using structured outputs, belief logging, escalation thresholds, and calibration monitoring. For portfolio managers, the empirical case study illustrates how LLM-inferred market regimes can be filtered, converted into exposures, and evaluated using standard performance, drawdown, VaR, CVaR, entropy, and Brier-score diagnostics.

\textbf{Practical deliverables.} The accompanying Python code downloads adjusted daily equity data from Massive.com, calls the OpenAI API for structured regime beliefs, applies Bayesian filtering, constructs an illustrative regime-sensitive portfolio policy, and generates figures plus LaTeX tables, including a complete numerical-experiment parameter table.
\end{minipage}}
\end{center}

\begin{abstract}
Agentic AI systems create model risk because uncertain beliefs are coupled to autonomous actions. This paper develops a mathematical framework for quantifying agentic AI risk by representing the system as a partially observed Markov decision process with latent states, Bayesian belief updates, control-dependent losses, and tail-risk functionals. The main methodological contribution is to treat a large language model as an uncertain semantic observation model: the LLM maps high-dimensional evidence into a probability vector over latent regimes, while a Bayesian filter imposes temporal coherence and produces auditable posterior beliefs. The resulting framework separates uncertainty quantification from risk measurement. Uncertainty is represented by posterior entropy, belief drift, and calibration error; risk is represented by the distribution of losses induced by decisions taken under those beliefs. The paper connects this construction to model risk management, coherent risk measures, Bayesian filtering, POMDP theory, robust control, and quantitative portfolio risk. An empirical case study using adjusted daily equity returns from Massive.com illustrates how LLM-inferred belief states can be combined with Bayesian filtering to produce regime probabilities, uncertainty diagnostics, calibration statistics, and VaR/CVaR-style risk measures. The framework is intended as a rigorous foundation for validating agentic AI in financial and other regulated decision environments.
\end{abstract}

\section{Introduction: background and literature review}

Model risk management has traditionally focused on the possibility that a model is misspecified, incorrectly implemented, poorly calibrated, or used outside the domain in which its assumptions are credible. The canonical regulatory formulation emphasizes adverse consequences from model errors and inappropriate model use \citep{sr117}. In quantitative finance, the same theme appears in valuation adjustment, stress testing, portfolio construction, coherent risk measures, and robust decision-making. A model is never merely a map from inputs to outputs; it is an instrument embedded in a decision process.

Agentic AI sharpens this point. A foundation model deployed as a chat interface may provide advice, explanations, or summaries. An agentic system can observe changing information, call tools, write files, trigger workflows, send messages, query private databases, or recommend trades. Consequently, agentic AI is naturally formulated as a closed-loop stochastic control problem. The mathematical language of partially observed Markov decision processes, Bayesian filtering, sequential decision theory, and risk-sensitive control is therefore more appropriate than a static benchmark-centered view of model quality \citep{puterman,kaelbling1998planning,bertsekas2017dynamic}.

The finance literature supplies several essential ingredients. Coherent and convex risk measures formalize the distinction between expected loss and tail-sensitive loss \citep{artzner1999coherent,foellmer2016stochastic}. Value-at-Risk and Conditional Value-at-Risk remain central in practice, even when their limitations are well understood \citep{jorion2007var,rockafellar2000optimization,mcneil2015quantitative}. Bayesian portfolio construction and Black--Litterman allocation show how subjective views can be combined with priors in a disciplined way \citep{black1992global,he2002intuition,satchell2000demystification}. Robust optimization and distributionally robust control address the fact that estimated probability laws are themselves uncertain \citep{hansen2008robust,ben2009robust,glasserman2014robust}. These ideas are directly relevant to LLM agents because the agent's belief state may be probabilistic, misspecified, prompt-sensitive, and regime-dependent.

Modern AI governance frameworks emphasize mapping, measuring, managing, and monitoring AI risks, including risks specific to generative AI such as confabulation, privacy leakage, information integrity, security, and value-chain integration \citep{nist2023airmf,nist2024genai}. Those frameworks are necessary, but many deployments still lack a mathematical bridge from qualitative governance categories to quantitative measures of residual risk. 

As shown in Dixon~\cite{dixon2026pomdp}, agentic AI systems may be formulated as partially observable Markov decision processes in which decisions are conditioned on posterior beliefs over latent environmental states. This paper builds on ~\cite{dixon2026pomdp} and provides such a bridge by treating the LLM as an observation model inside a Bayesian POMDP and by measuring how posterior uncertainty propagates into decisions and losses.

\subsection{Contributions}

The paper makes four contributions.

\begin{enumerate}[leftmargin=*]
\item \textbf{Theoretical contribution.} We define agentic AI model risk as a tail-sensitive loss functional over trajectories generated by actions taken under uncertain posterior beliefs. The framework distinguishes belief uncertainty from consequence-bearing risk and expresses residual risk as a function of posterior uncertainty, decision sensitivity, loss severity, autonomy amplification, and control effectiveness.
\item \textbf{Methodological contribution.} Following ~\cite{dixon2026pomdp}, we introduce an LLM-inferred Bayesian observation model. The LLM returns a structured probability vector over latent regimes; a tempered Bayesian filter combines that vector with a Markov transition prior. A novel variational characterization shows that the update is the unique belief closest to the transition prior while incorporating LLM evidence.
\item \textbf{Validation contribution.} We provide diagnostics for entropy, KL belief drift, Brier-score calibration, realized regime consistency, VaR, CVaR, drawdown, and value-of-information escalation. These diagnostics translate model risk management concepts into audit trails for agentic AI.
\item \textbf{Empirical contribution.} Finally, we provide reproducible Python code using Massive.com adjusted daily equity aggregate bars and the OpenAI API. The code generates regime-belief figures, uncertainty plots, risk-score plots, performance comparisons, and LaTeX tables, including all parameters used in the numerical experiment.
\end{enumerate}

\subsection{Overview of the rest of the paper}

Section~2 develops the POMDP and Bayesian filtering methodology. Section~3 specifies the role of LLMs as semantic observation models rather than autonomous decision makers. Section~4 explains how uncertainty becomes model risk, including calibration, prior assumptions, and the extent to which the construction is Bayesian. Section~5 presents the experimental setup and empirical results. Section~6 concludes with implications for governance, validation, and future research.

\section{Methodology}

Let $(\Omega,\mathcal F,\mathbb P)$ be a probability space and let time be discrete, $t=0,1,\ldots,T$. The environment has a hidden state $S_t\in\mathcal S=\{1,\ldots,K\}$. The agent observes $O_t\in\mathcal O$, forms a posterior belief $b_t\in\Delta_K$, chooses an action $A_t\in\mathcal A$, and incurs loss $\ell(S_t,A_t)$.

\begin{definition}[Agentic POMDP]
An agentic AI system is represented by
\[
\mathcal M=(\mathcal S,\mathcal A,\mathcal O,P,G,\pi,\ell,\gamma),
\]
where $P(s'\mid s,a)$ is the transition kernel, $G(o\mid s)$ is the observation kernel, $\pi(a\mid h_t)$ is the policy, $h_t=(o_0,a_0,\ldots,o_t)$ is the observable history, $\ell$ is a loss function, and $\gamma\in(0,1]$ is a discount factor.
\end{definition}

The belief state is the conditional distribution
\[
b_t(s)=\mathbb P(S_t=s\mid H_t),\qquad H_t=\sigma(O_0,A_0,\ldots,O_t).
\]
When $G$ is known, the standard Bayesian filter is
\[
\tilde b_t(s)=\sum_{s'}P(s\mid s',A_{t-1})b_{t-1}(s'),
\]
\[
b_t(s)=\frac{G(O_t\mid s)\tilde b_t(s)}{\sum_jG(O_t\mid j)\tilde b_t(j)}.
\]
The challenge in enterprise agentic AI is that $G$ is rarely known. Observations may include numerical features, documents, emails, tool outputs, market commentary, compliance logs, or natural-language instructions. A manually specified likelihood for such evidence is often impossible.

We therefore replace the unavailable likelihood with an LLM-inferred evidence vector. Let $q_t\in\Delta_K$ denote the LLM output, where $q_t(s)$ is the model's probability assessment for latent state $s$ after observing $O_t$. Given a transition matrix $P$ and an evidence-tempering parameter $\eta\in[0,1]$, define
\begin{equation}
\label{eq:llm_filter}
b_t(s)=\frac{q_t(s)^\eta\tilde b_t(s)}{\sum_jq_t(j)^\eta\tilde b_t(j)}.
\end{equation}

\begin{assumption}[Finite latent regimes]
The latent state space is finite and chosen before the experiment. The regimes are not learned ex post from realized returns; they are modelling hypotheses that encode the risk manager's prior partition of economically meaningful states.
\end{assumption}

\begin{assumption}[Markov transition prior]
The transition matrix $P$ is time-homogeneous over the experiment and diagonally dominant. This expresses persistence of regimes while allowing transitions to adjacent and stress states. It is a prior regularizer, not a claim that markets are literally stationary.
\end{assumption}

\begin{assumption}[Tempered semantic evidence]
The LLM output satisfies $q_t(s)>0$ and $\sum_s q_t(s)=1$. The parameter $\eta$ controls the influence of the LLM evidence. Values below one reduce overconfidence and acknowledge epistemic uncertainty in the LLM observation model.
\end{assumption}

\begin{proposition}[Variational form of the LLM-Bayesian update]
The belief update in Equation~\eqref{eq:llm_filter} is the unique minimizer over the simplex of
\[
\min_{b\in\Delta_K}\left\{D_{KL}(b\Vert\tilde b_t)-\eta\sum_s b(s)\log q_t(s)\right\}.
\]
\end{proposition}

\begin{proof}
The Lagrangian is
\[
\mathcal J(b,\lambda)=\sum_s b(s)\log\frac{b(s)}{\tilde b_t(s)}-\eta\sum_s b(s)\log q_t(s)+\lambda\left(\sum_sb(s)-1\right).
\]
The first-order condition implies
\[
\log b(s)-\log\tilde b_t(s)-\eta\log q_t(s)+1+\lambda=0.
\]
Thus $b(s)\propto\tilde b_t(s)q_t(s)^\eta$. Normalization yields Equation~\eqref{eq:llm_filter}. Strict convexity on the interior of the simplex gives uniqueness.
\end{proof}

\begin{algorithm}[h]
\caption{LLM-Bayesian belief filtering for agentic AI model risk}
\begin{algorithmic}[1]
\State Choose latent regimes $\mathcal S$, transition prior $P$, initial belief $b_0$, and evidence trust $\eta$.
\For{$t=1,\ldots,T$}
\State Observe evidence $O_t$.
\State Query the LLM for a structured probability vector $q_t\in\Delta_K$.
\State Compute the Markov-predicted prior $\tilde b_t=P^\top b_{t-1}$.
\State Update $b_t(s)\propto\tilde b_t(s)q_t(s)^\eta$.
\State Record entropy, belief drift, calibration diagnostics, action, realized loss, and controls.
\EndFor
\end{algorithmic}
\end{algorithm}

\section{LLMs as semantic observation models}

The LLM is used to infer latent belief states, not to choose trades or to determine final actions. This distinction is essential for governance. In the proposed architecture, the LLM is a probabilistic semantic sensor. It receives a controlled prompt containing observations and regime definitions, and it must return a schema-constrained probability vector. The downstream Bayesian filter, policy, and risk metrics operate on that vector.

Formally, let $\mathcal L_\theta$ denote the LLM and let $\Psi$ denote the prompt and schema. The observation model is
\[
q_t=\mathcal L_\theta(O_t;\Psi),\qquad q_t\in\Delta_K.
\]
In implementation, structured outputs are used so that the response contains exactly the required regime probabilities and a rationale field. This makes the LLM output machine-readable and audit-friendly. The rationale is not used in the Bayesian update; it is retained for validation and review.

The role of the LLM resembles an analyst supplying a probabilistic view in a Black--Litterman portfolio model. The view is not accepted mechanically. It is blended with a prior, tempered for uncertainty, and evaluated ex post. If $\eta=0$, the Bayesian system ignores the LLM and evolves only under the transition prior. If $\eta=1$, the LLM evidence is used without tempering. Intermediate values produce conservative evidence pooling.

Several risks remain. The LLM may be poorly calibrated, sensitive to prompt wording, or systematically biased by the feature representation. It may also understate ambiguity by returning overly concentrated probabilities. These risks motivate calibration diagnostics, entropy monitoring, prompt stability tests, and control thresholds. The point of the framework is not to assume that the LLM is Bayesian. It is to place the LLM inside a Bayesian validation wrapper that makes its uncertainty measurable and contestable.

\section{From uncertainty to model risk}

Uncertainty quantification and risk measurement are distinct. Uncertainty is a property of a belief distribution. Risk is a property of the loss distribution induced by acting under that belief. A highly uncertain agent can be low risk when consequences are small; a confident but misspecified agent can be high risk when consequences are severe.

\subsection{Posterior uncertainty and belief drift}

Posterior uncertainty is measured by normalized entropy
\[
\bar H(b_t)=-\frac{1}{\log K}\sum_{s=1}^Kb_t(s)\log b_t(s).
\]
Belief instability is measured by KL drift
\[
D_t=D_{KL}(b_t\Vert b_{t-1})=\sum_s b_t(s)\log\frac{b_t(s)}{b_{t-1}(s)}.
\]
High entropy indicates diffuse uncertainty across regimes. High drift indicates unstable beliefs, abrupt evidence changes, or regime transition. Both can be governance signals, but neither is a loss by itself.

\subsection{Calibration and model risk}

When ex-post labels $Y_t\in\{e_1,\ldots,e_K\}$ are available, calibration can be measured by the multiclass Brier score
\[
BS=\frac1T\sum_{t=1}^T\sum_{s=1}^K\left(b_t(s)-Y_t(s)\right)^2.
\]
Reliability diagrams, log scores, probability integral transforms for continuous targets, and regime-conditioned confusion matrices can also be used. In the empirical case study, realized regimes are heuristic labels derived from forward and trailing market behavior; they are not treated as ground truth about the true economy. They provide a validation proxy for whether beliefs align with subsequent market states.

Calibration matters because model risk arises when reported probabilities fail to match realized frequencies. In an agentic system, miscalibration is amplified by actions. Let $A_t=\pi(b_t,X_t)$, where $X_t$ denotes observed covariates. The trajectory loss is
\[
L_T=\sum_{t=0}^T\gamma^t\ell(S_t,A_t).
\]
Expected loss is useful but incomplete. Tail-sensitive loss is captured by
\[
\operatorname{VaR}_\alpha(L_T)=\inf\{x:\mathbb P(L_T\le x)\ge\alpha\},
\]
\[
\operatorname{CVaR}_\alpha(L_T)=\mathbb E[L_T\mid L_T\ge\operatorname{VaR}_\alpha(L_T)].
\]

\subsection{Is the framework fully Bayesian?}

The filtering layer is Bayesian conditional on the specified state space, transition matrix, evidence transformation, and evidence-tempering parameter. The entire framework is not fully Bayesian in the sense of placing priors on every unknown quantity, sampling the posterior over transition matrices, and integrating over LLM reliability parameters. In the base implementation, the transition matrix is a prior chosen ex ante, the regime set is fixed, and $\eta$ is a hyperparameter. A fully Bayesian extension would place priors such as
\[
P_{i\cdot}\sim\operatorname{Dirichlet}(\alpha_i),\qquad \eta\sim\operatorname{Beta}(a,b),
\]
and would integrate over posterior uncertainty in $P$ and $\eta$. The present version is therefore an empirical-Bayes or Bayesian-filtering framework with fixed hyperparameters. This is appropriate for a foundational validation study because it exposes the assumptions clearly and produces auditable quantities. It also avoids giving false precision to the LLM reliability model before sufficient validation data exist.

\subsection{Belief-at-Risk}
\label{subsec:belief_at_risk}

Uncertainty alone does not fully characterize model risk. A model can be uncertain about an inconsequential decision, or confident about a decision with severe downside consequences. The proposed local risk measure therefore combines three terms: posterior uncertainty, belief-state drift, and downside consequence.

Belief-state drift is measured by the Kullback--Leibler divergence
\begin{equation}
D_t
=
D_{KL}(b_t\Vert b_{t-1})
=
\sum_i b_{i,t}
\log
\left(
\frac{b_{i,t}}{b_{i,t-1}}
\right).
\end{equation}
The proposed local AI risk measure, termed \emph{Belief-at-Risk}, is
\begin{equation}
\boxed{
BaR_t
=
\widetilde H_t
\left(
1+D_{KL}(b_t\Vert b_{t-1})
\right)
CVaR_{0.95,t}
}
\label{eq:belief_at_risk}
\end{equation}
where \(CVaR_{0.95,t}\) is the rolling downside consequence measure.

\subsection{AI Governance: Control-Adjusted Agentic Risk}

The purpose of this section is to preliminary attempt to bridge the Belief-at-Risk measure developed in the previous section to AI governance. The Belief-at-Risk measure quantifies the intrinsic model risk generated by the agent. In practice, however, agentic AI systems are rarely deployed without safeguards. Human approval workflows, tool restrictions, monitoring systems, audit trails, and policy enforcement mechanisms are specifically designed to reduce the likelihood and severity of adverse outcomes.

Consequently, the quantity of greatest interest to risk managers is not the intrinsic risk of the agent, but the residual risk that remains after controls have been applied.

To formalize this idea, let $R_{\mathrm{agent}}$ denote the intrinsic risk generated by the agent. In the present framework this may be identified with the proposed Belief-at-Risk measure,

\[
R_{\mathrm{agent}}
=
BaR_t.
\]

Suppose a collection of governance controls is implemented. Let $\kappa \in [0,1]$ denote the overall effectiveness of the control framework, where

\begin{itemize}
\item \(\kappa=0\) corresponds to no meaningful controls,
\item \(\kappa=1\) corresponds to perfect mitigation,
\item \(0<\kappa<1\) corresponds to partial mitigation.
\end{itemize}

Residual risk is then defined by

\begin{equation}
R_{\mathrm{residual}}
=
(1-\kappa)
R_{\mathrm{agent}}.
\label{eq:residual_risk}
\end{equation}

The interpretation is straightforward. A control framework that is estimated to eliminate 80\% of risk corresponds to \(\kappa=0.8\), leaving only 20\% of the original risk.

In practice, \(\kappa\) can be estimated using historical incident data. If a control reduces the frequency of adverse outcomes from \(N_{\mathrm{before}}\) incidents to \(N_{\mathrm{after}}\) incidents, a simple estimator is

\begin{equation}
\kappa
=
1
-
\frac{N_{\mathrm{after}}}
{N_{\mathrm{before}}}.
\end{equation}

For example, if a human approval process reduces unauthorized actions from 100 incidents per year to 20 incidents per year, the estimated control effectiveness is

\[
\kappa
=
1-\frac{20}{100}
=
0.8.
\]

The residual risk is therefore
\[
R_{\mathrm{residual}}
=
0.2
R_{\mathrm{agent}}.
\]
In larger organizations, multiple categories of risk may be considered separately. Let \(R_i\) denote the intrinsic risk associated with category \(i\), such as model risk, privacy risk, compliance risk, or cybersecurity risk. Let \(\lambda_i\) denote the relative importance of that category, where

\[
\sum_i \lambda_i = 1.
\]

A practical calibration is obtained by allocating weights according to historical losses:
\begin{equation}
\lambda_i
=
\frac{L_i}
{\sum_j L_j},
\end{equation}

where \(L_i\) denotes cumulative historical losses attributable to category \(i\). The resulting residual risk measure becomes

\begin{equation}
R_{\mathrm{residual}}
=
\sum_i
(1-\kappa_i)
\lambda_i
R_i.
\end{equation}

This formulation provides a simple and practical bridge between quantitative model risk measurement and AI governance. The intrinsic risk generated by the agent is first quantified using Belief-at-Risk. Governance controls are then incorporated through empirically estimated control effectiveness parameters, producing a residual risk measure that can be monitored and managed using familiar enterprise risk-management practices.

\section{Experimental Setup and Results}
\label{sec:experiments}

The purpose of this empirical study is to investigate whether uncertainty estimates derived from an LLM-enhanced Bayesian filtering framework can be transformed into meaningful quantitative measures of model risk. The analysis focuses on four related objects:

\begin{enumerate}
\item latent regime beliefs,
\item posterior uncertainty,
\item belief-state instability,
\item the resulting Belief-at-Risk (BaR) measure.
\end{enumerate}

The experimental design, numerical assumptions, and model parameters are summarized in Table~\ref{tab:experiment_parameters}.

\input{tables/experiment_parameters.tex}

Several aspects of Table~\ref{tab:experiment_parameters} are noteworthy. First, the latent state space consists of four economically interpretable regimes: Risk-On, Neutral, Risk-Off, and Crisis. These latent regimes correspond to different levels of risk appetite in financial markets. Risk-On states are associated with positive expected returns and lower perceived risk, Neutral states with balanced market conditions, Risk-Off states with elevated uncertainty and defensive positioning, and Crisis states with extreme stress characterized by rapid repricing, increased volatility, and flight-to-quality behavior. The regimes are deliberately defined at a high level to provide an interpretable latent-state representation that can be inferred from heterogeneous market and macroeconomic observations.

Second, a uniform prior is assumed over latent states in order to avoid introducing directional information into the initialization procedure. Third, downside consequence is measured using a rolling 95\% Conditional Value-at-Risk estimate computed over a 60-day window. Finally, the belief-state portfolio construction maps latent states into economically meaningful portfolio exposures ranging from fully long to fully defensive positions.

\subsection{Belief-State Portfolio Construction}

The portfolio is constructed directly from the inferred belief state. Let

\begin{equation}
b_t
=
\begin{bmatrix}
P_t(\text{RiskOn})\\
P_t(\text{Neutral})\\
P_t(\text{RiskOff})\\
P_t(\text{Crisis})
\end{bmatrix}
\end{equation}
denote the Bayesian-filtered posterior distribution over latent market states. The state-dependent exposure function is defined by

\begin{align}
w(\text{RiskOn}) &= 1,\\
w(\text{Neutral}) &= 0.5,\\
w(\text{RiskOff}) &= -0.5,\\
w(\text{Crisis}) &= -1.
\end{align}
These weights arise naturally from an ordinal utility function over latent regimes. The specific values are not important, however what matters is:
$w(\text{RiskOn})>w(\text{Neutral})>w(\text{RiskOff})>w(\text{Crisis})$. The above choice of weights is simply a convenient normalization of this ordering.
The resulting portfolio exposure is

\begin{equation}
w_t
=
P_t(\text{RiskOn})
+
0.5P_t(\text{Neutral})
-
0.5P_t(\text{RiskOff})
-
P_t(\text{Crisis})
\label{eq:belief_weight}
\end{equation}
which corresponds to the posterior expectation of the state-dependent exposure function. Portfolio returns are generated according to

\begin{equation}
r_t^{agent}
=
w_{t-1}
r_t^{SPY}.
\label{eq:agent_return}
\end{equation}
Equations~(\ref{eq:belief_weight}) and (\ref{eq:agent_return}) create a direct mapping from latent-state uncertainty to observable economic outcomes. This link is essential because the objective of the paper is to measure model risk arising from uncertainty in the inferred belief state rather than forecasting performance alone.

\subsection{Posterior Regime Dynamics}
Figure~\ref{fig:regime_probabilities} illustrates the Bayesian-filtered posterior probabilities associated with the four latent regimes.

\begin{figure}[htbp]
\centering
\includegraphics[width=\textwidth]{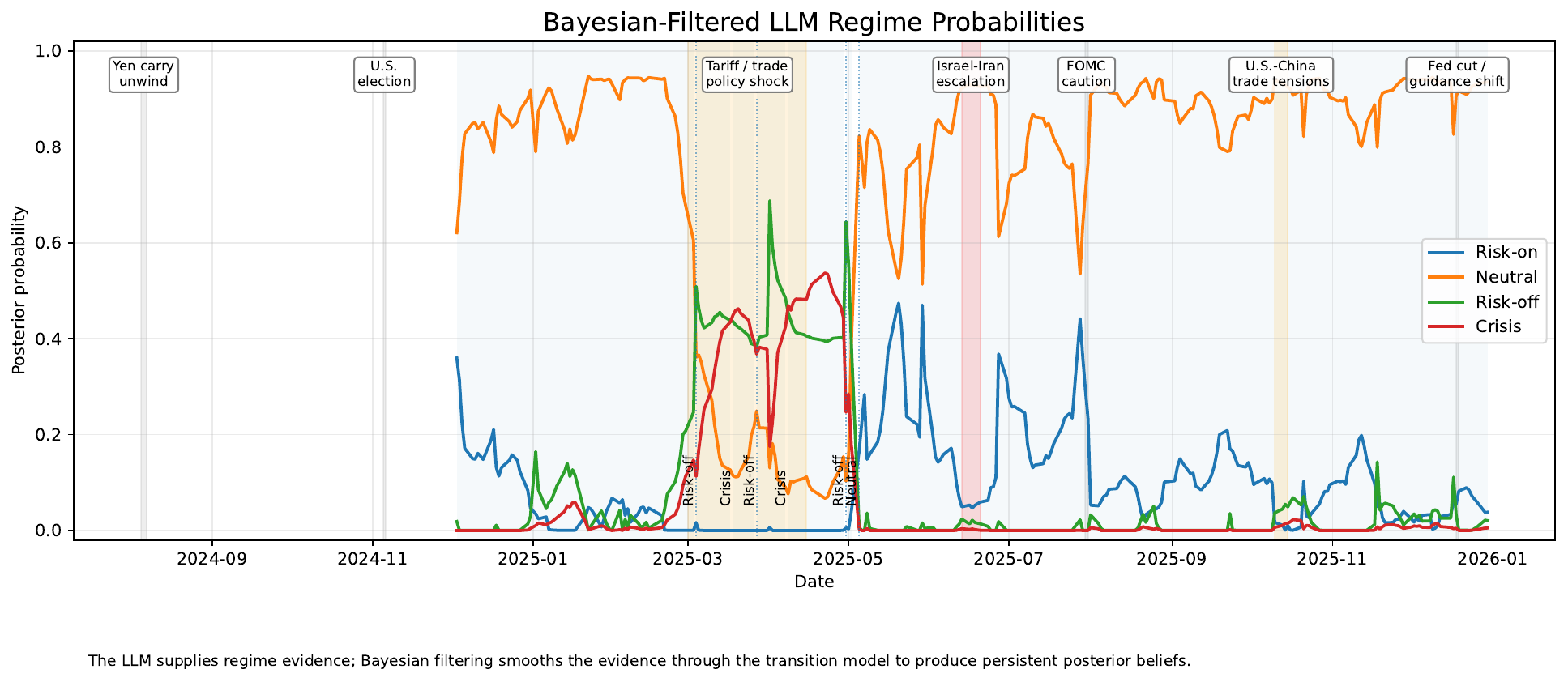}
\caption{
Bayesian-filtered posterior probabilities over the latent regime set. Event windows indicate major macroeconomic and geopolitical developments. The Bayesian filter transforms noisy LLM-generated regime assessments into persistent posterior beliefs.
}
\label{fig:regime_probabilities}
\end{figure}

Several observations emerge from Figure~\ref{fig:regime_probabilities}. First, the Neutral regime dominates much of the sample, indicating that the filtering framework does not over-classify market stress. Second, substantial reallocations of probability mass occur during the March--April 2025 tariff-policy uncertainty episode. During this period both Risk-Off and Crisis probabilities rise materially while the Neutral regime declines. Third, posterior Crisis probabilities again increase during the June 2025 Israel--Iran escalation before gradually reverting as uncertainty dissipates. The resulting regime dynamics are economically plausible and broadly consistent with the chronology of observed macro-financial events.

The statistical properties of the posterior beliefs are summarized in Table~\ref{tab:regime_summary}.

\input{tables/regime_summary.tex}

Table~\ref{tab:regime_summary} confirms the visual evidence from Figure~\ref{fig:regime_probabilities}. Neutral states occur most frequently, while Risk-Off and Crisis states exhibit lower average probabilities but substantially larger excursions during stress periods. This behaviour is consistent with the intended interpretation of these regimes as low-frequency, high-impact states.

\subsection{Posterior Entropy and Uncertainty Quantification}

Figure~\ref{fig:posterior_entropy} reports the normalized posterior entropy process.

\begin{figure}[htbp]
\centering
\includegraphics[width=\textwidth]{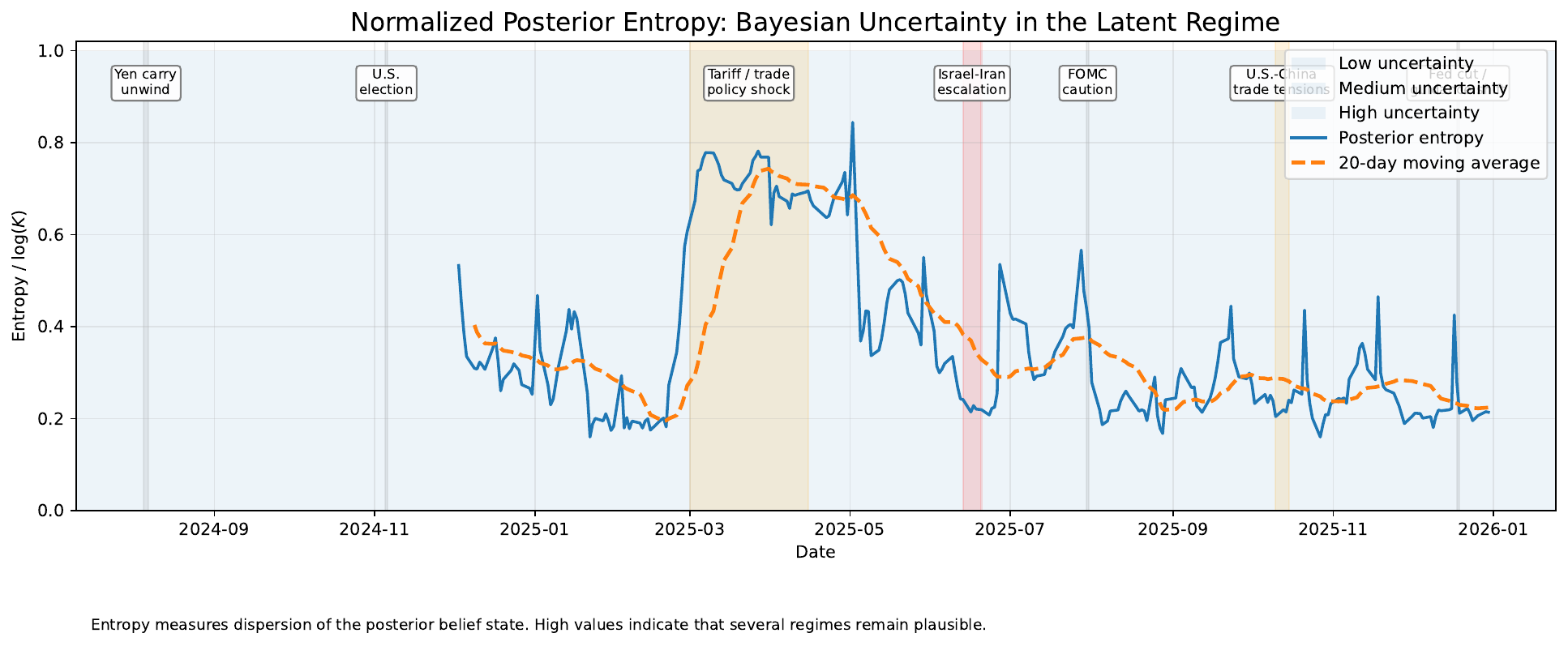}
\caption{
Normalized posterior entropy of the Bayesian-filtered belief state. High values indicate that multiple latent regimes remain plausible, while low values indicate concentrated posterior beliefs.
}
\label{fig:posterior_entropy}
\end{figure}

Entropy provides a direct measure of uncertainty regarding the latent market state. The figure reveals that uncertainty rises substantially during the tariff-policy shock period. Entropy remains elevated throughout March and April 2025 and subsequently declines as the posterior beliefs become increasingly concentrated.

The behavior of Figure~\ref{fig:posterior_entropy} is particularly important because uncertainty is distinct from market direction. A portfolio may have a small net exposure while remaining highly uncertain regarding the underlying latent regime. Consequently, entropy captures information unavailable from portfolio positions alone.

\subsection{Belief-at-Risk Dynamics}

Figure~\ref{fig:belief_at_risk} presents the proposed Belief-at-Risk measure.

\begin{figure}[htbp]
\centering
\includegraphics[width=\textwidth]{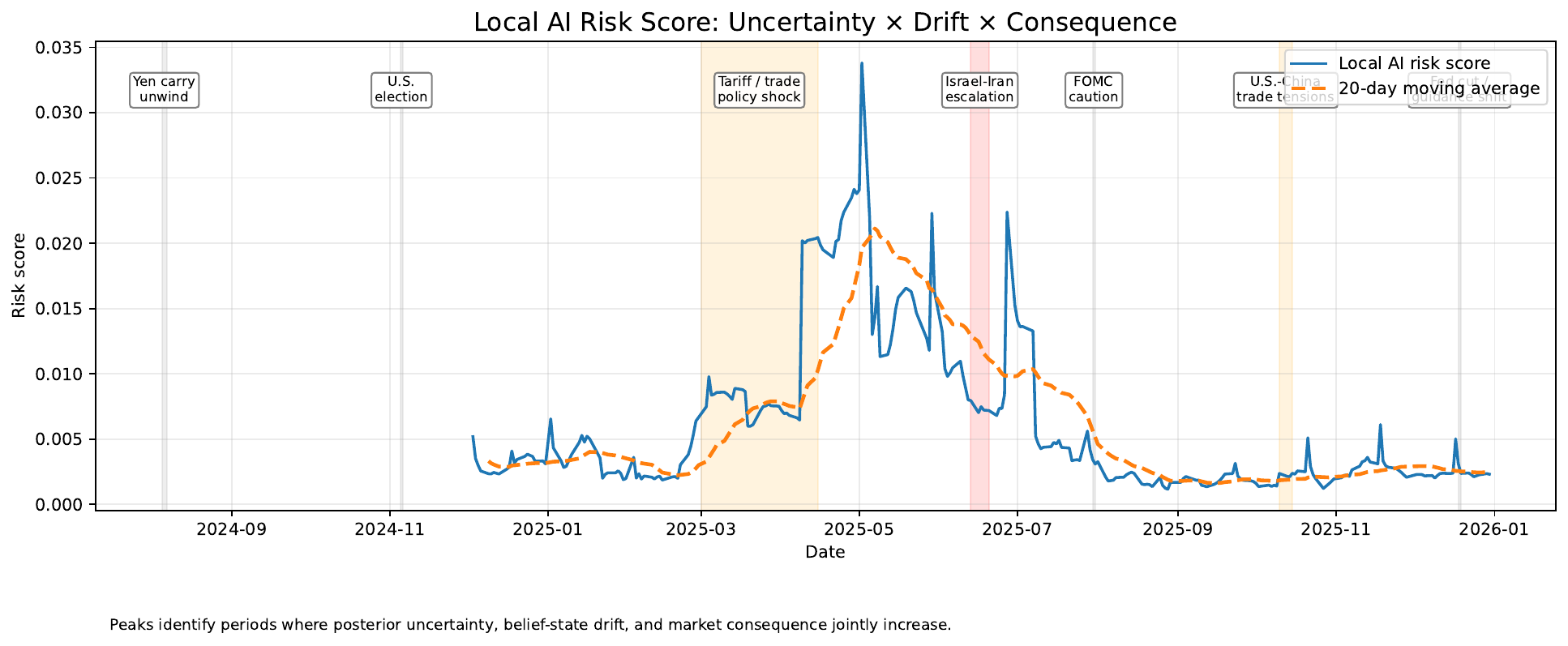}
\caption{
Belief-at-Risk (BaR) dynamics. The measure combines posterior uncertainty, belief-state instability, and downside economic consequence. Peaks occur when all three components increase simultaneously.
}
\label{fig:belief_at_risk}
\end{figure}

The largest Belief-at-Risk values occur during the March--June 2025 interval. This period coincides with elevated entropy, substantial regime transitions, and increased downside economic consequence. Importantly, the measure peaks after uncertainty emerges and before uncertainty is fully resolved. This behavior is consistent with the interpretation of Belief-at-Risk as a model-risk measure rather than a market-risk measure.

Following the June 2025 geopolitical shock, the risk measure declines steadily despite continued market fluctuations. This result highlights an important distinction between volatility and model risk. Market volatility may remain elevated while model risk falls if uncertainty and belief instability decline.

\subsection{Performance Diagnostics}

The economic consequences of the belief-state portfolio are summarized in Table~\ref{tab:performance_metrics}.

\input{tables/performance_metrics.tex}

The purpose of Table~\ref{tab:performance_metrics} is not to evaluate investment performance per se. Rather, the portfolio provides a controlled mechanism through which uncertainty can be translated into measurable economic outcomes. The resulting statistics therefore quantify the financial consequences of acting upon the inferred belief state.

\subsection{Risk Diagnostics}

Table~\ref{tab:risk_diagnostics} summarizes the principal components of the proposed Belief-at-Risk framework.

\input{tables/risk_diagnostics.tex}

Table~\ref{tab:risk_diagnostics} reveals that uncertainty, instability, and economic consequence contribute complementary information. Entropy measures the uncertainty of the posterior belief state, KL divergence quantifies instability in the evolution of beliefs, and Conditional Value-at-Risk captures downside economic consequence. Elevated model risk emerges primarily when all three components increase simultaneously.

\subsection{Summary of Findings}

Taken together, Figures~\ref{fig:regime_probabilities}--\ref{fig:belief_at_risk} and Tables~\ref{tab:experiment_parameters}--\ref{tab:risk_diagnostics} support three principal conclusions.
First, Bayesian filtering transforms noisy LLM-generated regime assessments into persistent and economically interpretable latent-state beliefs. Second, uncertainty quantification provides information that is distinct from both portfolio exposure and realized returns. Posterior entropy identifies periods in which multiple latent explanations remain plausible even when portfolio positions appear stable.

Finally, the proposed Belief-at-Risk measure successfully integrates uncertainty, belief instability, and economic consequence into a unified quantitative measure of model risk. The resulting framework provides a natural bridge between Bayesian decision theory, POMDP-based agent architectures, and quantitative model risk management.

\section{Conclusion}

Agentic AI model risk is consequence-bearing uncertainty in a closed-loop decision system. This paper provides a rigorous framework for quantifying that risk by combining LLM-inferred belief states with Bayesian filtering, calibration diagnostics, and loss-based risk measures. The framework separates the semantic inference role of the LLM from the decision and governance layers. It therefore gives CROs, technology leaders, portfolio managers, and model validators a language for measuring uncertainty, testing calibration, defining escalation thresholds, and controlling residual agentic AI risk.

Future work should estimate transition matrices hierarchically, place priors on LLM reliability, compare multiple models as competing observation kernels, incorporate richer macroeconomic and textual observations, and study distributionally robust policies under adversarial prompt and market regimes. The broader message is that trusted agentic AI should not be governed only by policies and checklists. It should be measured with the same mathematical seriousness used for financial model risk.

\bibliographystyle{plainnat}

\end{document}

%% file: tables/experiment_parameters.tex
\begin{table}
\caption{Numerical specifications and modelling assumptions used in the LLM--Bayesian filtering experiment.}
\label{tab:experiment_parameters}
\begin{tabular}{lll}
\toprule
Experimental quantity & Mathematical notation & Numerical specification \\
\midrule
Sample start date & Start date & 2021-01-01\\
Sample end date & End date & 2025-12-31 \\
Equity universe & $\mathcal{A}$ & AAPL, MSFT, GOOGL, AMZN, JPM, SPY \\
Massive.com sleep interval & $\Delta_{\mathrm{API}}$ & 10 seconds \\
Adjusted equity prices & Adjusted close & True \\
Return frequency & $\Delta t$ & Daily \\
Number of latent regimes & $K$ & 4 \\
Latent regime set & $\mathcal{S}$ & Risk-on, Neutral, Risk-off, Crisis \\
Initial regime prior & $b_0$ & (0.25, 0.25, 0.25, 0.25) \\
Bayesian filter floor & $\varepsilon$ & 1e-12 \\
Rolling CVaR window & $W$ & 60 trading days \\
CVaR confidence level & $\alpha$ & 95\% \\
Risk-on exposure & $w(RO)$ & 1 \\
Neutral exposure & $w(N)$ & 0.5 \\
Risk-off exposure & $w(RF)$ & -0.5 \\
Crisis exposure & $w(C)$ & -1 \\
\bottomrule
\end{tabular}
\end{table}

%% file: tables/regime_summary.tex
\begin{table}[h!]
\caption{Summary statistics for Bayesian-filtered LLM posterior regime probabilities.}
\label{tab:regime_summary}
\resizebox{\textwidth}{!}{%
\begin{tabular}{lccccc}
\toprule
Latent regime & Mean posterior probability & Posterior standard deviation & Minimum posterior probability & Maximum posterior probability & Dominant-regime frequency \\
\midrule
Risk-on & 0.0952 & 0.0998 & 7.37e-05 & 0.475 & 0 \\
Neutral & 0.748 & 0.271 & 0.0674 & 0.948 & 0.84 \\
Risk-off & 0.089 & 0.162 & 0.000431 & 0.688 & 0.0781 \\
Crisis & 0.0677 & 0.149 & 2.26e-08 & 0.537 & 0.0818 \\
\bottomrule
\end{tabular}
}
\end{table}

%% file: tables/performance_metrics.tex
\begin{table}
\caption{Out-of-sample performance diagnostics for the belief-state portfolio induced by the LLM--Bayesian filter.}
\label{tab:performance_metrics}
\begin{tabular}{lll}
\toprule
Performance statistic & Mathematical symbol & Estimated value \\
\midrule
Annualized return & $\mu_{\mathrm{ann}}$ & 0.044 \\
Annualized volatility & $\sigma_{\mathrm{ann}}$ & 0.107 \\
Sharpe ratio & $SR$ & 0.411 \\
Maximum drawdown & $MDD$ & -0.0899 \\
Daily value-at-risk, 95\% & $VaR_{0.95}$ & -0.00725 \\
Daily conditional value-at-risk, 95\% & $CVaR_{0.95}$ & -0.0159 \\
\bottomrule
\end{tabular}
\end{table}

%% file: tables/risk_diagnostics.tex
\begin{table}[!h]
\caption{Uncertainty, belief-instability and consequence diagnostics for the proposed local AI risk measure.}
\label{tab:risk_diagnostics}
\begin{tabular}{lll}
\toprule
Risk diagnostic & Mathematical symbol & Estimated value \\
\midrule
Mean normalized posterior entropy & $\bar{H}/\log K$ & 0.363 \\
Maximum normalized posterior entropy & $\max_t H_t/\log K$ & 0.844 \\
Mean belief-state drift & $\overline{D}_{KL}(b_t\Vert b_{t-1})$ & 0.0235 \\
Maximum belief-state drift & $\max_t D_{KL}(b_t\Vert b_{t-1})$ & 0.819 \\
Mean local AI risk score & $\bar{\mathcal{R}}$ & 0.00602 \\
Maximum local AI risk score & $\max_t \mathcal{R}_t$ & 0.0338 \\
Mean rolling CVaR consequence & $\overline{CVaR}_{0.95}$ & 0.0148 \\
Maximum rolling CVaR consequence & $\max_t CVaR_{0.95,t}$ & 0.0327 \\
\bottomrule
\end{tabular}
\end{table}

%% file: llm_bayesian_agentic_ai_risk.bbl
\begin{thebibliography}{99}

\bibitem[Artzner et~al.(1999)]{artzner1999coherent}
Artzner, P., Delbaen, F., Eber, J., and Heath, D. (1999).
Coherent measures of risk.
\emph{Mathematical Finance}, 9(3), 203--228.

\bibitem[Ben-Tal et~al.(2009)]{ben2009robust}
Ben-Tal, A., El Ghaoui, L., and Nemirovski, A. (2009).
\emph{Robust Optimization}.
Princeton University Press.

\bibitem[Bertsekas(2017)]{bertsekas2017dynamic}
Bertsekas, D. (2017).
\emph{Dynamic Programming and Optimal Control}.
Athena Scientific.

\bibitem[Black and Litterman(1992)]{black1992global}
Black, F. and Litterman, R. (1992).
Global portfolio optimization.
\emph{Financial Analysts Journal}, 48(5), 28--43.

\bibitem[Cover and Thomas(2006)]{cover2006elements}
Cover, T. and Thomas, J. (2006).
\emph{Elements of Information Theory}.
Wiley.

\bibitem[Diebold et~al.(1998)]{diebold1998evaluating}
Diebold, F., Gunther, T., and Tay, A. (1998).
Evaluating density forecasts with applications to financial risk management.
\emph{International Economic Review}, 39(4), 863--883.

\bibitem{dixon2026pomdp}
Dixon, M. (2026).
\emph{Agentic AI as a Partially Observable Markov Decision Process}.
Working Paper, Quiota Research.


\bibitem[Federal Reserve and OCC(2011)]{sr117}
Board of Governors of the Federal Reserve System and Office of the Comptroller of the Currency (2011).
Supervisory guidance on model risk management, SR 11-7 / OCC 2011-12.

\bibitem[Föllmer and Schied(2016)]{foellmer2016stochastic}
Föllmer, H. and Schied, A. (2016).
\emph{Stochastic Finance: An Introduction in Discrete Time}.
De Gruyter.

\bibitem[Glasserman(2004)]{glasserman2004monte}
Glasserman, P. (2004).
\emph{Monte Carlo Methods in Financial Engineering}.
Springer.

\bibitem[Glasserman and Xu(2014)]{glasserman2014robust}
Glasserman, P. and Xu, X. (2014).
Robust risk measurement and model risk.
\emph{Quantitative Finance}, 14(1), 29--58.

\bibitem[Hansen and Sargent(2008)]{hansen2008robust}
Hansen, L. and Sargent, T. (2008).
\emph{Robustness}.
Princeton University Press.

\bibitem[He and Litterman(2002)]{he2002intuition}
He, G. and Litterman, R. (2002).
The intuition behind Black--Litterman model portfolios.
Goldman Sachs Asset Management working paper.

\bibitem[Jorion(2007)]{jorion2007var}
Jorion, P. (2007).
\emph{Value at Risk: The New Benchmark for Managing Financial Risk}.
McGraw-Hill.

\bibitem[Kaelbling et~al.(1998)]{kaelbling1998planning}
Kaelbling, L., Littman, M., and Cassandra, A. (1998).
Planning and acting in partially observable stochastic domains.
\emph{Artificial Intelligence}, 101(1--2), 99--134.

\bibitem[McNeil et~al.(2015)]{mcneil2015quantitative}
McNeil, A., Frey, R., and Embrechts, P. (2015).
\emph{Quantitative Risk Management: Concepts, Techniques and Tools}.
Princeton University Press.

\bibitem[NIST(2023)]{nist2023airmf}
National Institute of Standards and Technology (2023).
\emph{Artificial Intelligence Risk Management Framework (AI RMF 1.0)}.
NIST AI 100-1.

\bibitem[NIST(2024)]{nist2024genai}
National Institute of Standards and Technology (2024).
\emph{Artificial Intelligence Risk Management Framework: Generative Artificial Intelligence Profile}.
NIST AI 600-1.

\bibitem[OpenAI(2026)]{openai2026structured}
OpenAI (2026).
Structured outputs API documentation.

\bibitem[Puterman(1994)]{puterman}
Puterman, M. (1994).
\emph{Markov Decision Processes: Discrete Stochastic Dynamic Programming}.
Wiley.

\bibitem[Rockafellar and Uryasev(2000)]{rockafellar2000optimization}
Rockafellar, R. and Uryasev, S. (2000).
Optimization of Conditional Value-at-Risk.
\emph{Journal of Risk}, 2, 21--41.

\bibitem[Satchell and Scowcroft(2000)]{satchell2000demystification}
Satchell, S. and Scowcroft, A. (2000).
A demystification of the Black--Litterman model.
\emph{Journal of Asset Management}, 1, 138--150.

\end{thebibliography}
